\documentclass[11pt]{article}

\usepackage[margin=1in]{geometry}
\usepackage{amsmath,amssymb,amsthm,mathtools}
\usepackage{algorithm}
\usepackage{algpseudocode}
\usepackage{booktabs}
\usepackage{enumitem}
\usepackage[hidelinks]{hyperref}
\usepackage[font=small,labelfont=bf]{caption}

\newtheorem{theorem}{Theorem}
\newtheorem{lemma}[theorem]{Lemma}

\newtheorem{definition}[theorem]{Definition}

\newcommand{\OPT}{\operatorname{OPT}}

\newcommand{\R}{\mathbb{R}}
\newcommand{\deltaout}{\delta^+}
\newcommand{\RPMNC}{\textsc{RPMNC}}

\title{Three-Terminal Reachability-Preserving Minimum Node Cut:\\
Planar Hardness and a General-Graph \(O(\sqrt n)\)-Approximation}

\author{Qi Duan, Carnegie Mellon University}
\date{}

\begin{document}

\maketitle

\begin{abstract}
We study the three-terminal reachability-preserving minimum node cut problem
(\RPMNC).  The input is an undirected graph \(G=(V,E)\), nonnegative vertex
weights on nonterminal vertices, two protected terminals \(s_1,s_2\), and a
target terminal \(t\).  The goal is to delete a minimum-weight set of
nonterminal vertices so that \(t\) is disconnected from the protected terminals,
while \(s_1\) and \(s_2\) remain connected.  This problem captures a basic
``separate while preserve'' requirement that arises in biological intervention
design, image analysis with connectivity constraints, and cyber-security attack
graph mitigation, where deleting or blocking a node represents preventing the
corresponding action, state, or biological entity from participating in a harmful
pathway.

We prove two results.  First, the weighted planar version of three-terminal
\RPMNC{} is NP-complete.  The reduction is from \textsc{Independent Set} on
3-regular Hamiltonian planar graphs and uses a one-sided blocker construction.
Second, we give a polynomial-time \(O(\sqrt n)\)-approximation algorithm for
general graphs.  The algorithm is based on an exact path--separator identity,
a directed split-graph representation of rooted vertex separators, and a
root-linear approximation of a monotone submodular separator function.
\end{abstract}

\section{Introduction}

Minimum cut problems are central in graph algorithms and combinatorial
optimization.  The classical \(s\)-\(t\) cut problem asks for the cheapest set
of edges, or vertices, whose removal separates two terminals.  In many
applications, however, separation alone is not enough.  One often wants to
isolate a harmful or unwanted target while preserving the connectivity of a
trusted or functional part of the network.  This leads naturally to
connectivity-preserving or reachability-preserving cut problems.

In this paper we study the following three-terminal node-cut problem.  Given an
undirected graph \(G=(V,E)\), two protected terminals \(s_1,s_2\), and a target
terminal \(t\), delete a minimum-weight set of nonterminal vertices \(X\) such
that, in \(G-X\), the protected terminals remain connected while \(t\) is
disconnected from them.  We call this problem three-terminal
reachability-preserving minimum node cut, abbreviated as \RPMNC.

The problem is motivated by several domains.

\paragraph{Biology and medicine.}
Biological systems are frequently represented as graphs or networks whose nodes
are genes, proteins, metabolites, cell states, or higher-level biological
entities.  In pathway intervention, one may want to block a disease-driving
target while preserving a functional pathway between beneficial or required
entities.  A node deletion can model inhibiting a protein, suppressing a gene,
removing a metabolite, or perturbing an interaction state.  The preservation
constraint is important: an intervention that disconnects the harmful target but
also destroys a required biological route may be clinically or biologically
undesirable.

\paragraph{Image processing and computer vision.}
Graph cuts are widely used in image segmentation.  Pixels, superpixels, or
regions are modeled as graph vertices, and cut objectives separate foreground
from background or separate competing labels.  Standard cuts do not necessarily
preserve connectivity among selected regions.  In segmentation tasks where two
trusted seeds or anatomical landmarks must remain connected while a target
region is separated, a reachability-preserving node cut provides a natural
combinatorial abstraction.

\paragraph{Cyber security and attack graphs.}
Attack graphs encode how an attacker can chain actions, vulnerabilities, and
system states to reach a goal.  In this setting, deleting or blocking a node in
the graph corresponds to deploying a mitigation that prevents the corresponding
attack action or state from occurring.  A reachability-preserving cut asks for a
minimum-cost set of mitigations that blocks the attacker's path to a target
while preserving connectivity among benign, operational, or monitoring states.

\paragraph{Contribution.}
We prove the following two results.

\begin{itemize}[leftmargin=2em]
    \item We show that weighted planar three-terminal \RPMNC{} is NP-complete.
    The reduction is from \textsc{Independent Set} on 3-regular Hamiltonian
    planar graphs, which is NP-complete~\cite{FleischnerSabidussiSarvanov2010}.

    \item We give a polynomial-time \(O(\sqrt n)\)-approximation algorithm for
    three-terminal \RPMNC{} in general graphs.  The algorithm reduces the
    problem to minimizing a rooted vertex-separator function over
    \(s_1\)-\(s_2\) paths.  This separator function becomes a directed rooted
    cut function after vertex splitting and is therefore monotone and
    submodular.  We then use a root-linear approximation, in the spirit of
    the submodular function approximation framework of Goemans, Harvey, Iwata,
    and Mirrokni~\cite{GoemansHarveyIwataMirrokni2009}.
\end{itemize}

\section{Related Work}

The classical max-flow/min-cut theorem gives polynomial-time algorithms for
ordinary \(s\)-\(t\) cut problems~\cite{FordFulkerson1962,AhujaMagnantiOrlin1993}.
Node cuts can be reduced to edge cuts by the standard vertex-splitting
transformation.

Connectivity-preserving cuts were studied by Duan and Xu under the name
connectivity preserving minimum cut~\cite{DuanXu2014}.  Their work introduced
node-cut and edge-cut variants and established strong hardness results for
general connectivity-preserving node cuts, as well as polynomial-time
algorithms for some planar edge-cut cases.  The present work focuses on the
three-terminal node-cut case with two protected terminals and one target
terminal.

Multiway cut is another important generalization of minimum cut.  Dahlhaus et
al.~\cite{DahlhausJohnsonPapadimitriouSeymourYannakakis1994} showed that
multiway cut becomes NP-hard for three terminals.  Planar multiway cut has also
been studied extensively, including approximation schemes for planar
instances~\cite{BateniHajiaghayiKleinMathieu2011}.  However, multiway cut
separates terminals from each other, while \RPMNC{} separates one target from
a protected connected pair.

Vertex separator problems and their approximability have a long history.
Feige, Hajiaghayi, and Lee~\cite{FeigeHajiaghayiLee2008} developed improved
approximation algorithms for minimum-weight vertex separators.  In contrast,
\RPMNC{} is not a balanced separator problem; it is a cut--preserve problem
with a specified terminal pair whose connectivity must survive.

Our approximation algorithm uses standard tools from submodular optimization.
Polymatroids and their greedy structure are classical; see, for example,
Fujishige~\cite{Fujishige2005}.  The root-linear approximation used here is
closely related to the general problem of approximating monotone submodular
functions everywhere~\cite{GoemansHarveyIwataMirrokni2009}.

Graph-cut methods have also been important in computer vision, especially for
image segmentation~\cite{BoykovJolly2001,BoykovKolmogorov2004}.  In biology,
graph and network models are widely used to represent molecular and cellular
systems~\cite{KoutrouliKaratzasPaezEspinoPavlopoulos2020}.  In cyber security,
attack graphs represent chained attacks and are used for security analysis and
mitigation planning~\cite{MellHarang2015}.

\section{Problem Definition}

Let \(G=(V,E)\) be an undirected graph.  Let
\[
s_1,s_2,t\in V
\]
be three distinct terminals.  The terminals are undeletable.  Each nonterminal
vertex
\[
v\in V\setminus\{s_1,s_2,t\}
\]
has a nonnegative weight \(w(v)\).

\begin{definition}[Three-terminal \RPMNC]
A feasible \RPMNC{} solution is a set
\[
X\subseteq V\setminus\{s_1,s_2,t\}
\]
such that, in \(G-X\),
\[
s_1 \leftrightarrow s_2
\qquad\text{and}\qquad
s_1 \not\leftrightarrow t.
\]
Since \(s_1\) and \(s_2\) remain connected, the second condition is equivalent
to saying that \(t\) is disconnected from both protected terminals.  The
objective is to minimize
\[
w(X)=\sum_{v\in X}w(v).
\]
\end{definition}

Equivalently, a feasible solution leaves a connected component \(C\) of
\(G-X\) such that
\[
s_1,s_2\in C,
\qquad
t\notin C.
\]

\section{NP-Completeness of Weighted Planar Three-Terminal \RPMNC}
\label{sec:planar-hardness}

We prove that weighted planar three-terminal \RPMNC{} is NP-complete.

\subsection{Source Problem}

We reduce from \textsc{Independent Set} on 3-regular Hamiltonian planar graphs.
This problem is NP-complete~\cite{FleischnerSabidussiSarvanov2010}.

The input is a 3-regular Hamiltonian planar graph
\[
H=(U,F),
\qquad |U|=n,
\]
together with an integer \(K\).  The question is whether \(H\) has an
independent set of size at least \(K\).

Let
\[
C=(u_1,u_2,\ldots,u_n,u_1)
\]
be a Hamiltonian cycle of \(H\), given with a planar embedding.  Set the
\RPMNC{} budget to
\[
B=n-K.
\]
Let
\[
M=n+1.
\]
All auxiliary routing, wall, and connector vertices that should never be
deleted by a budget-\(B\) solution are assigned weight \(M\).

\subsection{Construction}

For every source vertex \(u\in U\), create a selector vertex
\[
o_u
\]
with weight
\[
w(o_u)=1.
\]

For every source edge \(e\in F\) incident with \(u\), create a gate vertex
\[
g_{u,e}
\]
with weight
\[
w(g_{u,e})=0.
\]

The construction contains three terminals
\[
s_1=a,
\qquad
s_2=b,
\qquad
t=z.
\]
There is a connected \(z\)-side rail containing \(z\).  For every
\(u\in U\), add a protected path from the \(z\)-side rail to \(o_u\).  All
internal vertices on these rail and routing paths have weight \(M\).  For
every edge \(e\ni u\), add the edge
\[
o_u g_{u,e}.
\]

For every source edge
\[
e=uv\in F,
\]
create an edge certificate gadget \(D_e\) with two parallel branches:
\[
p_e-g_{u,e}-q_e
\]
and
\[
p_e-g_{v,e}-q_e.
\]
Thus \(D_e\) is traversable after deleting a set \(X\) if and only if at
least one of the two endpoint gates
\[
g_{u,e},\qquad g_{v,e}
\]
survives.

Finally, order the source edges arbitrarily as
\[
e_1,e_2,\ldots,e_m.
\]
Connect the edge gadgets in series from \(a\) to \(b\):
\[
a \to D_{e_1}\to D_{e_2}\to \cdots \to D_{e_m}\to b.
\]
Every internal connector vertex has weight \(M\).  Therefore any
\(a\)-to-\(b\) path through the certificate chain must pass through every edge
gadget \(D_e\).

\subsection{Planarity}

\begin{lemma}
\label{lem:hardness-planarity}
The constructed \RPMNC{} instance is planar and has size polynomial in
\(|H|\).
\end{lemma}

\begin{proof}
Start with the given planar Hamiltonian embedding of \(H\).  Place small
pairwise disjoint vertex boxes for \(u_1,\ldots,u_n\) along the Hamiltonian
cycle.  The selector \(o_u\) is placed inside the box for \(u\).  The
\(z\)-side rail is drawn along one side of the Hamiltonian corridor, and the
protected paths from this rail to the selectors are drawn inside the
corresponding vertex boxes.

For each source edge \(e=uv\), use a sufficiently thin neighborhood of the
embedded edge \(e\) in the planar drawing of \(H\).  Place the gate
\(g_{u,e}\) near the \(u\)-end of this corridor and the gate \(g_{v,e}\) near
the \(v\)-end.  The two branches
\[
p_e-g_{u,e}-q_e
\qquad\text{and}\qquad
p_e-g_{v,e}-q_e
\]
are drawn inside this thin corridor.  Since the original edge corridors of
the planar embedding of \(H\) are pairwise noncrossing except at their
endpoints, these certificate gadgets can be drawn without crossings.

The serial connectors between edge gadgets are drawn using a thin regular
neighborhood of a planar connector tree in the complement of the already drawn
local gadget interiors.  The lanes realize the serial order
\[
D_{e_1},D_{e_2},\ldots,D_{e_m}
\]
without creating graph-theoretic shortcuts: only the intended consecutive
connector edges are included in the graph.  Therefore the constructed graph is
planar and has polynomial size.
\end{proof}

\subsection{Correctness of the Reduction}

Suppose \(H\) has an independent set
\[
S\subseteq U
\]
with
\[
|S|\ge K.
\]
Construct a deletion set \(X_S\) as follows.  If
\[
u\notin S,
\]
delete the selector \(o_u\).  If
\[
u\in S,
\]
delete all gates incident with \(u\):
\[
g_{u,e}\in X_S
\qquad
\text{for every } e\ni u.
\]
The weight is
\[
w(X_S)
=
|\{u\in U:u\notin S\}|
=
n-|S|
\le
n-K
=
B,
\]
because all gates have weight zero.

For each vertex gadget \(u\), all \(z\)-to-certificate channels through that
gadget are blocked.  If \(u\notin S\), the selector \(o_u\) is deleted.  If
\(u\in S\), then all gates \(g_{u,e}\), \(e\ni u\), are deleted.  Hence no
route from the \(z\)-side rail can enter the certificate chain.

Now consider any source edge
\[
e=uv.
\]
Because \(S\) is independent, at least one endpoint of \(e\) is not in \(S\).
Assume without loss of generality that \(u\notin S\).  Then \(o_u\) is
deleted, but the gate \(g_{u,e}\) is not deleted.  Hence the branch
\[
p_e-g_{u,e}-q_e
\]
survives in \(D_e\).  Therefore every edge gadget in the serial chain is
traversable, and \(a\) is connected to \(b\) in \(G-X_S\).  Thus \(X_S\) is
feasible.

Conversely, suppose \(X\) is a feasible \RPMNC{} solution with
\[
w(X)\le B=n-K.
\]
Since all high-weight auxiliary vertices have weight
\[
M=n+1>B,
\]
the set \(X\) contains no high-weight rail, routing, wall, or connector
vertex.  Define
\[
S=\{u\in U:o_u\notin X\}.
\]

We first prove the blocker property.  Let \(u\in S\).  Then
\[
o_u\notin X.
\]
We claim that every incident gate of \(u\) must be deleted:
\[
g_{u,e}\in X
\qquad
\text{for every } e\ni u.
\]
Suppose not.  Then for some edge \(e\ni u\),
\[
g_{u,e}\notin X.
\]
Since \(X\) is feasible, \(a\) and \(b\) are connected in \(G-X\).  The
certificate graph is a serial chain, so every edge gadget must be traversable.
In particular, \(D_e\) is traversable.  Since \(g_{u,e}\notin X\), the branch
\[
p_e-g_{u,e}-q_e
\]
survives, and \(g_{u,e}\) lies in the surviving \(a\)-to-\(b\) component.

On the other hand, \(z\) is connected to \(o_u\) through the surviving
\(z\)-side rail, and \(o_u\) is adjacent to \(g_{u,e}\).  Since neither
\(o_u\) nor \(g_{u,e}\) is deleted, this gives a path from \(z\) to the
\(a\)-component in \(G-X\), contradicting feasibility.  Therefore all gates
incident with every \(u\in S\) are deleted.

Only selectors have positive weight among the low-weight vertices.  Hence
\[
w(X)\ge |\{u\in U:o_u\in X\}|.
\]
Therefore
\[
|\{u\in U:o_u\in X\}|
\le
w(X)
\le
n-K.
\]
Since \(S=\{u:o_u\notin X\}\), we get
\[
|S|
=
n-|\{u:o_u\in X\}|
\ge K.
\]

It remains to prove that \(S\) is independent.  Suppose, for contradiction,
that there is an edge
\[
e=uv\in F
\]
with
\[
u\in S
\qquad\text{and}\qquad
v\in S.
\]
By the blocker property,
\[
g_{u,e}\in X
\qquad\text{and}\qquad
g_{v,e}\in X.
\]
Both branches of \(D_e\) are blocked, so \(D_e\) is not traversable.  Since
the certificate chain is serial, this disconnects \(a\) from \(b\),
contradicting feasibility.  Thus \(S\) is an independent set of size at least
\(K\).

\begin{theorem}
\label{thm:planar-rpmnc-npcomplete}
The decision version of weighted planar three-terminal \RPMNC{} with
nonnegative vertex weights is NP-complete.
\end{theorem}

\begin{proof}
The reduction above is polynomial and preserves planarity.  It shows that
\[
H \text{ has an independent set of size at least } K
\]
if and only if the constructed planar \RPMNC{} instance has a feasible
solution of weight at most
\[
B=n-K.
\]
Thus the problem is NP-hard.

The problem is in NP because, given a deletion set \(X\), one can check in
polynomial time whether
\[
s_1\leftrightarrow s_2
\qquad\text{and}\qquad
s_1\not\leftrightarrow t
\]
hold in \(G-X\), and whether \(w(X)\le B\).  Hence the decision version is
NP-complete.
\end{proof}

\section{An \(O(\sqrt n)\)-Approximation for General Graphs}
\label{sec:approx}

In this section we prove that three-terminal undirected \RPMNC{} admits a
polynomial-time \(O(\sqrt n)\)-approximation in general graphs.

For a set
\[
A\subseteq V\setminus\{t\},
\]
define
\[
f_N(A)
=
\min
\left\{
w(X):
X\subseteq V\setminus(A\cup\{t\}),
\ X \text{ separates } t \text{ from every vertex of } A
\right\}.
\]
If no such separator exists, set
\[
f_N(A)=+\infty.
\]

\begin{lemma}[Path--separator identity]
\label{lem:path-separator}
For every feasible three-terminal \RPMNC{} instance,
\[
\OPT_{\RPMNC}
=
\min_{P:s_1\leadsto s_2,\ t\notin V(P)}
f_N(V(P)),
\]
where the minimum is over all \(s_1\)-to-\(s_2\) paths avoiding \(t\).
\end{lemma}

\begin{proof}
Let \(X^\star\) be an optimal \RPMNC{} solution.  In \(G-X^\star\), let
\(C^\star\) be the connected component containing \(s_1\) and \(s_2\).
Then \(t\notin C^\star\).  Since \(C^\star\) is connected, it contains an
\(s_1\)-to-\(s_2\) path \(P^\star\), and this path avoids \(t\).

Because \(X^\star\) separates \(t\) from the whole component \(C^\star\), it
separates \(t\) from every vertex of \(V(P^\star)\).  Hence
\[
f_N(V(P^\star))
\le
w(X^\star)
=
\OPT_{\RPMNC}.
\]

Conversely, let \(P\) be any \(s_1\)-to-\(s_2\) path avoiding \(t\), and let
\(X\) realize \(f_N(V(P))\).  Since
\[
X\subseteq V\setminus(V(P)\cup\{t\}),
\]
the path \(P\) remains intact in \(G-X\), so \(s_1\) and \(s_2\) remain
connected.  Since \(X\) separates \(t\) from every vertex of \(V(P)\), it
separates \(t\) from both protected terminals.  Therefore \(X\) is feasible.
Taking the minimum over all such paths proves the identity.
\end{proof}

\subsection{Directed Split-Graph Representation}

Construct a directed split graph \(D\) as follows.  For every vertex
\(v\in V\), create two copies \(v^{-}\) and \(v^{+}\).  Add an arc
\[
v^{-}\to v^{+}.
\]
If \(v\) is deletable, this arc has capacity \(w(v)\).  If
\[
v\in\{s_1,s_2,t\},
\]
this arc has capacity \(+\infty\).  For every undirected edge
\[
\{u,v\}\in E,
\]
add two infinite-capacity arcs
\[
u^{+}\to v^{-}
\qquad\text{and}\qquad
v^{+}\to u^{-}.
\]
Let the root be
\[
r=t^{+}.
\]
For
\[
A\subseteq V\setminus\{t\},
\]
define
\[
A^{-}=\{a^{-}:a\in A\}.
\]

\begin{lemma}[Split-graph equivalence]
\label{lem:split-graph}
For every
\[
A\subseteq V\setminus\{t\},
\]
we have
\[
f_N(A)
=
\min
\left\{
c_D(\deltaout_D(S)):
r\in S,\ S\cap A^{-}=\emptyset
\right\}.
\]
\end{lemma}

\begin{proof}
A finite \(r\)-to-\(A^{-}\) cut in \(D\) can use only arcs of the form
\[
v^{-}\to v^{+}
\]
corresponding to deletable vertices \(v\).  All arcs encoding adjacency and
all terminal-splitting arcs have infinite capacity.  Hence every finite
directed cut corresponds to a set
\[
X\subseteq V\setminus(A\cup\{t\})
\]
of deleted nonterminal vertices, with the same total weight.

A path from \(t\) to a vertex \(a\in A\) in \(G-X\) corresponds exactly to a
directed path from \(t^{+}\) to \(a^{-}\) in \(D\) after removing the arcs
\[
x^{-}\to x^{+}
\qquad
(x\in X).
\]
Thus the directed cut separates \(r=t^{+}\) from \(A^{-}\) if and only if
\(X\) separates \(t\) from every vertex of \(A\) in the original graph.
\end{proof}

\begin{lemma}[Monotonicity and submodularity]
\label{lem:submodular}
The function \(f_N\) is normalized, monotone, and submodular on
\(V\setminus\{t\}\).
\end{lemma}

\begin{proof}
Normalization and monotonicity are immediate.  For submodularity, let \(S_A\)
and \(S_B\) be optimal source-side sets for \(A\) and \(B\), respectively.
Then \(S_A\cap S_B\) is feasible for \(A\cup B\), and \(S_A\cup S_B\) is
feasible for \(A\cap B\).  Since directed cut capacity is submodular,
\[
c_D(\deltaout_D(S_A))
+
c_D(\deltaout_D(S_B))
\ge
c_D(\deltaout_D(S_A\cap S_B))
+
c_D(\deltaout_D(S_A\cup S_B)).
\]
The claim follows.
\end{proof}

\subsection{Root-Linear Approximation}

Let
\[
R=V\setminus\{t\},
\qquad
m=|R|\le n.
\]
For each
\[
v\in R,
\]
define
\[
\lambda_v=f_N(\{v\}).
\]
Consider the polymatroid
\[
P(f_N)
=
\left\{
y\in\R_{\ge 0}^{R}:
y(A)\le f_N(A)\ \text{for all } A\subseteq R
\right\}.
\]
Let \(y^\star\) be a proportional-fair point of \(P(f_N)\), namely
\[
y^\star
\in
\arg\max_{y\in P(f_N),\ y_v>0\ \forall v}
\sum_{v\in R}\log y_v.
\]
Define
\[
a_v=\lambda_v y^\star_v.
\]

\begin{lemma}[Root-linear approximation]
\label{lem:root-linear}
For every \(A\subseteq R\),
\[
\sqrt{a(A)}
\le
f_N(A)
\le
\sqrt m\,\sqrt{a(A)}.
\]
\end{lemma}

\begin{proof}
Since \(y^\star\in P(f_N)\), we have
\[
y^\star(A)\le f_N(A).
\]
Also, by monotonicity,
\[
\lambda_v=f_N(\{v\})\le f_N(A)
\qquad
\text{for every } v\in A.
\]
Therefore
\[
a(A)
=
\sum_{v\in A}\lambda_v y^\star_v
\le
f_N(A)\sum_{v\in A}y^\star_v
=
f_N(A)y^\star(A)
\le
f_N(A)^2.
\]
Taking square roots gives the lower bound.

For the upper bound, the first-order optimality condition for \(y^\star\)
implies that for every \(x\in P(f_N)\),
\[
\sum_{v\in R}\frac{x_v}{y^\star_v}\le m.
\]
By the greedy characterization of polymatroids, for every \(A\subseteq R\)
there exists \(x^A\in P(f_N)\), supported on \(A\), such that
\[
x^A(A)=f_N(A)
\]
and
\[
0\le x^A_v\le \lambda_v.
\]
By Cauchy's inequality,
\[
f_N(A)
=
\sum_{v\in A}x^A_v
\le
\sqrt{a(A)}
\sqrt{
\sum_{v\in A}
\frac{(x^A_v)^2}{\lambda_v y^\star_v}
}.
\]
Since \(x^A_v\le \lambda_v\),
\[
\frac{(x^A_v)^2}{\lambda_v y^\star_v}
\le
\frac{x^A_v}{y^\star_v}.
\]
Thus
\[
f_N(A)
\le
\sqrt{a(A)}
\sqrt{
\sum_{v\in A}
\frac{x^A_v}{y^\star_v}
}
\le
\sqrt{a(A)}\sqrt m.
\]
\end{proof}

\subsection{Separation Oracle}

Given
\[
y\in\R_{\ge 0}^{R},
\]
we need to test whether
\[
y(A)\le f_N(A)
\qquad
\text{for every } A\subseteq R.
\]
Using the split graph \(D\), add a new sink \(q\).  For each \(v\in R\),
add an arc
\[
v^{-}\to q
\]
of capacity \(y_v\).  Let
\[
Y=\sum_{v\in R}y_v.
\]
Then
\[
\max_{A\subseteq R}\{y(A)-f_N(A)\}
=
Y-
\min_{S:r\in S,\ q\notin S}
\left(
c_D(\deltaout_D(S))
+
\sum_{v:v^{-}\in S}y_v
\right).
\]
Thus separation over \(P(f_N)\) reduces to one directed minimum cut.

\subsection{Approximation Algorithm}

\begin{algorithm}[H]
\caption{\(O(\sqrt n)\)-Approximation for Three-Terminal \RPMNC}
\label{alg:rpmnc-sqrtn}
\begin{algorithmic}[1]
\Require Undirected graph \(G=(V,E)\), vertex weights \(w\), terminals \(s_1,s_2,t\).
\Ensure A feasible \RPMNC{} node cut \(X\).
\State Build the directed split graph \(D\).
\State Define the rooted separator function \(f_N\).
\State Compute \(\lambda_v=f_N(\{v\})\) for every relevant \(v\in V\setminus\{t\}\).
\State Compute a proportional-fair point \(y^\star\in P(f_N)\).
\State Set \(a_v=\lambda_v y^\star_v\).
\State Find a shortest \(s_1\)-to-\(s_2\) path \(P\) in \(G-t\) using vertex lengths \(a_v\).
\State Compute a minimum node separator \(X\) separating \(t\) from every vertex of \(V(P)\).
\State \Return \(X\).
\end{algorithmic}
\end{algorithm}

\begin{theorem}
\label{thm:rpmnc-sqrtn}
Three-terminal undirected \RPMNC{} admits a polynomial-time
\(O(\sqrt n)\)-approximation in general graphs.
\end{theorem}

\begin{proof}
Let \(P\) be the path chosen by Algorithm~\ref{alg:rpmnc-sqrtn}.  Let
\(X^\star\) be an optimal \RPMNC{} solution, and let \(P^\star\) be an
\(s_1\)-to-\(s_2\) path contained in the connected component of
\(G-X^\star\) containing \(s_1\) and \(s_2\).  By Lemma~\ref{lem:path-separator},
\[
f_N(V(P^\star))
\le
\OPT_{\RPMNC}.
\]
Since \(P\) is shortest with respect to the modular vertex lengths \(a_v\),
\[
a(V(P))
\le
a(V(P^\star)).
\]
Using Lemma~\ref{lem:root-linear}, we obtain
\[
f_N(V(P))
\le
\sqrt m\sqrt{a(V(P))}
\le
\sqrt m\sqrt{a(V(P^\star))}
\le
\sqrt m\,f_N(V(P^\star)).
\]
Therefore
\[
f_N(V(P))
\le
\sqrt n\,\OPT_{\RPMNC}.
\]
The algorithm outputs a minimum node separator \(X\) separating \(t\) from
every vertex of \(V(P)\).  Since \(X\cap V(P)=\emptyset\), the path \(P\)
remains in \(G-X\), so \(s_1\) and \(s_2\) remain connected.  Since \(X\)
separates \(t\) from every vertex of \(V(P)\), it separates \(t\) from both
protected terminals.  Thus \(X\) is feasible and has weight at most
\[
\sqrt n\,\OPT_{\RPMNC}.
\]
\end{proof}

\section{Discussion}

The NP-completeness proof relies on the NP-completeness of
\textsc{Independent Set} on 3-regular Hamiltonian planar graphs
\cite{FleischnerSabidussiSarvanov2010}.  The approximation proof relies on
standard max-flow/min-cut, vertex splitting, polymatroid, and submodular
approximation tools
\cite{FordFulkerson1962,AhujaMagnantiOrlin1993,Fujishige2005,
GoemansHarveyIwataMirrokni2009}.

To the best of our knowledge, the specific combination proved here---weighted
planar NP-completeness of three-terminal \RPMNC{} and an \(O(\sqrt n)\)
approximation for the same node-cut problem in general graphs---does not appear
as an official published result in the existing connectivity-preserving cut
literature.  The closest official prior work is the JCSS paper of Duan and
Xu~\cite{DuanXu2014}, which studies connectivity-preserving minimum cut
variants.

\section{Conclusion}

We studied the three-terminal reachability-preserving minimum node cut problem.
The problem asks for a minimum-cost set of nonterminal vertices whose deletion
separates a target terminal \(t\) from two protected terminals \(s_1,s_2\),
while preserving the connectivity between the protected terminals.

We proved that the weighted planar version is NP-complete.  This shows that
the node-cut version remains computationally difficult even in planar graphs.
We also gave a polynomial-time \(O(\sqrt n)\)-approximation for general graphs.
The approximation algorithm is based on an exact path--separator identity and a
root-linear approximation of a monotone submodular rooted vertex-separator
function.

Several directions remain open.  The most immediate is to improve the
approximation ratio for planar graphs.  Another direction is to study
unweighted or strictly positive weighted variants, directed attack-graph
variants, and practical heuristics that exploit the shortest-path and separator
structure of the approximation algorithm.

\bibliographystyle{plain}
\bibliography{main}

@book{FordFulkerson1962,
  author    = {Lester R. Ford and Delbert R. Fulkerson},
  title     = {Flows in Networks},
  publisher = {Princeton University Press},
  year      = {1962}
}

@book{AhujaMagnantiOrlin1993,
  author    = {Ravindra K. Ahuja and Thomas L. Magnanti and James B. Orlin},
  title     = {Network Flows: Theory, Algorithms, and Applications},
  publisher = {Prentice Hall},
  year      = {1993}
}

@article{DuanXu2014,
  author  = {Qi Duan and Jinhui Xu},
  title   = {On the Connectivity Preserving Minimum Cut Problem},
  journal = {Journal of Computer and System Sciences},
  volume  = {80},
  number  = {4},
  pages   = {837--848},
  year    = {2014},
  doi     = {10.1016/j.jcss.2014.01.003}
}

@article{FleischnerSabidussiSarvanov2010,
  author  = {Herbert Fleischner and Gert Sabidussi and Vladimir I. Sarvanov},
  title   = {Maximum Independent Sets in 3- and 4-Regular Hamiltonian Graphs},
  journal = {Discrete Mathematics},
  volume  = {310},
  number  = {20},
  pages   = {2742--2749},
  year    = {2010},
  doi     = {10.1016/j.disc.2010.05.028}
}

@inproceedings{GoemansHarveyIwataMirrokni2009,
  author    = {Michel X. Goemans and Nicholas J. A. Harvey and Satoru Iwata and Vahab S. Mirrokni},
  title     = {Approximating Submodular Functions Everywhere},
  booktitle = {Proceedings of the Twentieth Annual ACM-SIAM Symposium on Discrete Algorithms},
  series    = {SODA 2009},
  pages     = {535--544},
  year      = {2009},
  publisher = {SIAM},
  doi       = {10.1137/1.9781611973068.59}
}

@book{Fujishige2005,
  author    = {Satoru Fujishige},
  title     = {Submodular Functions and Optimization},
  edition   = {2},
  series    = {Annals of Discrete Mathematics},
  volume    = {58},
  publisher = {Elsevier},
  year      = {2005}
}

@article{DahlhausJohnsonPapadimitriouSeymourYannakakis1994,
  author  = {Elias Dahlhaus and David S. Johnson and Christos H. Papadimitriou and Paul D. Seymour and Mihalis Yannakakis},
  title   = {The Complexity of Multiterminal Cuts},
  journal = {SIAM Journal on Computing},
  volume  = {23},
  number  = {4},
  pages   = {864--894},
  year    = {1994},
  doi     = {10.1137/S0097539792225297}
}

@inproceedings{BateniHajiaghayiKleinMathieu2011,
  author    = {MohammadHossein Bateni and MohammadTaghi Hajiaghayi and Philip N. Klein and Claire Mathieu},
  title     = {A Polynomial-Time Approximation Scheme for Planar Multiway Cut},
  booktitle = {Proceedings of the Twenty-Second Annual ACM-SIAM Symposium on Discrete Algorithms},
  series    = {SODA 2011},
  pages     = {639--655},
  publisher = {SIAM},
  year      = {2011},
  doi       = {10.1137/1.9781611973082.50}
}

@article{FeigeHajiaghayiLee2008,
  author  = {Uriel Feige and MohammadTaghi Hajiaghayi and James R. Lee},
  title   = {Improved Approximation Algorithms for Minimum-Weight Vertex Separators},
  journal = {SIAM Journal on Computing},
  volume  = {38},
  number  = {2},
  pages   = {629--657},
  year    = {2008},
  doi     = {10.1137/05064299X}
}

@inproceedings{BoykovJolly2001,
  author    = {Yuri Y. Boykov and Marie-Pierre Jolly},
  title     = {Interactive Graph Cuts for Optimal Boundary and Region Segmentation of Objects in N-D Images},
  booktitle = {Proceedings of the Eighth IEEE International Conference on Computer Vision},
  series    = {ICCV 2001},
  volume    = {1},
  pages     = {105--112},
  year      = {2001},
  doi       = {10.1109/ICCV.2001.937505}
}

@article{BoykovKolmogorov2004,
  author  = {Yuri Boykov and Vladimir Kolmogorov},
  title   = {An Experimental Comparison of Min-Cut/Max-Flow Algorithms for Energy Minimization in Vision},
  journal = {IEEE Transactions on Pattern Analysis and Machine Intelligence},
  volume  = {26},
  number  = {9},
  pages   = {1124--1137},
  year    = {2004},
  doi     = {10.1109/TPAMI.2004.60}
}

@article{KoutrouliKaratzasPaezEspinoPavlopoulos2020,
  author  = {Mikaela Koutrouli and Evangelos Karatzas and David Paez-Espino and Georgios A. Pavlopoulos},
  title   = {A Guide to Conquer the Biological Network Era Using Graph Theory},
  journal = {Frontiers in Bioengineering and Biotechnology},
  volume  = {8},
  pages   = {34},
  year    = {2020},
  doi     = {10.3389/fbioe.2020.00034}
}

@inproceedings{MellHarang2015,
  author    = {Peter Mell and Richard Harang},
  title     = {Minimizing Attack Graph Data Structures},
  booktitle = {Proceedings of the Tenth International Conference on Software Engineering Advances},
  series    = {ICSEA 2015},
  pages     = {30--37},
  year      = {2015}
}

\end{document}